\documentclass{amsart}
\newtheorem{theorem}{Theorem}[section]

\theoremstyle{definition}

\newtheorem{definition}{Definition}
\newtheorem{corollary}{Corollary}

\newtheorem{example}{Example}
\theoremstyle{remark}
\newtheorem{remark}[theorem]{Remark}
\usepackage{ragged2e}
\usepackage{amsmath}
\numberwithin{equation}{section}
\usepackage[table, svgnames, dvipsnames]{xcolor}
\usepackage{makecell, cellspace, caption}
\usepackage{tabularx, ragged2e, booktabs, caption}
\usepackage{amssymb}
\usepackage{amsmath}
\usepackage{color,soul}
\usepackage{enumitem}
\usepackage{tabularx,ragged2e,booktabs,caption}
\usepackage{longtable}
\usepackage{graphicx}
\usepackage[colorlinks]{hyperref}

\begin{document}
	\title{Classification of LCD and self-dual codes over a finite non-unital local ring}

	\author{Anup Kushwaha}
    \address{Department of Mathematics, Indian Institute of Technology Patna, Patna-801106}
    \curraddr{}
    \email{ anup$\textunderscore$2221ma11@iitp.ac.in}
    \thanks{}

   \author{Indibar Debnath}
    \address{Department of Mathematics, Indian Institute of Technology Patna, Patna-801106}
    \curraddr{}
    \email{ indibar\_1921ma07@iitp.ac.in}
    \thanks{}

   \author{Om Prakash}
    \address{Department of Mathematics, Indian Institute of Technology Patna, Patna-801106}
    \curraddr{}
    \email{om@iitp.ac.in}
    \thanks{}

   \author{ Patrick Sol\'e$^*$}
   \address{I2M, (CNRS, Aix-Marseille University, Centrale Marseille) Marseille, France.}
   \curraddr{}
   \email{sole@enst.fr}
   \thanks{* Corresponding author}
	
	
	
	\dedicatory{}

\begin{abstract}
This work explores LCD and self-dual codes over a noncommutative non-unital ring $ E_p= \langle r,s ~|~ pr =ps=0,~ r^2=r,~ s^2=s,~ rs=r,~ sr=s \rangle$ of order $p^2$ where $p$ is a prime. Initially, we study the monomial equivalence of two free $E_p$-linear codes. In addition, a necessary and sufficient condition is derived for a free $E_p$-linear code to be MDS and almost MDS (shortly AMDS). Then, we use these results to classify MDS and AMDS LCD codes over $E_2$ and $E_3$ under monomial equivalence for lengths up to $6$. Subsequently, we study left self-dual codes over the ring $E_p$ and classify MDS and AMDS left self-dual codes over $E_2$ and $E_3$ for lengths up to $12$. Finally, we study self-dual codes over the ring $E_p$ and classify MDS and AMDS self-dual codes over $E_2$ and $E_3$ for short lengths.
\end{abstract}

\maketitle

\section{Introduction}
 In traditional studies of coding theory, finite fields have been used as code alphabets, and codes have a structure of vector spaces over these fields. Over the past three decades, considerable research has extended to codes as modules over various rings, primarily concentrating on unital rings mostly commutative. However, recent attention has shifted significantly towards non-unital rings  \cite{Alah22,Alah23,Kim2022a,Kim2022b,Shi21}. In this paper, as an alphabet, we consider a noncommutative non-unital ring $E_p=\langle r,s ~|~ pr =ps=0,~ r^2=r,~ s^2=s,~ rs=r,~ sr=s \rangle$ of order $p^2$ present in the notation of Fine \cite{Fine93}, where $p$ is a prime number.\par
To provide some context, first recall some fundamental concepts. A code $C$ of length $n$ and size $M$ is called an ($n, M, d$)-code if its minimum (Hamming) distance is $d$. For an ($n, M, d$)-code $C$ over an alphabet of size $q$, the Singleton bound provides an upper bound for the code size. It states that $M\leq q^{n-d+1}$. The linear codes that achieve this bound are considered optimal linear codes and known as maximum distance separable (MDS) codes. In other words, if $M=q^{n-d+1}$, the corresponding code is called MDS. Further, if $M=q^{n-d}$, the corresponding code is known as an almost MDS (AMDS) code. The distance of a linear code quantifies its ability to detect and correct errors. Thus, MDS and AMDS codes hold significant importance in theoretical and practical applications, as their error-detecting and error-correcting capacities are maximum. These facts show that MDS and AMDS codes are worth studying. \par
The dual of a linear code $C$ over a finite field $\mathbb{F}_q$ is the collection of all orthogonal vectors to $C$ under Euclidean inner product, and denoted by $C^\perp$. If a linear code $C$ meets its dual trivially, i.e., $C \cap C^\perp= \{\boldsymbol{0}\}$, it is known as a linear complementary dual (LCD) code. The class of LCD codes over finite fields was first introduced by  Massey in 1992 in \cite{Massey92}, as a solution to an Information theory problem. Then they were studied by Sendrier in connection with code automorphism algorithms \cite{Send04}. In 2016, Carlet and Guilley \cite{Carl16}  explored their applications in Boolean masking against Fault-Injection and Side-Channel attacks on embarked hardware. Further, LCD codes are also crucial in secret-sharing schemes \cite{Yadav}. These applications of LCD codes have reignited interest in constructing LCD codes with large distances. Therefore, the construction of MDS and AMDS LCD codes holds significant theoretical and practical importance.\par

  In 2015, Liu and Liu \cite{Liu15} provided certain conditions under which a linear code over finite chain rings becomes LCD. In 2019, Liu and Wang \cite{Zih19} studied LCD codes over finite commutative rings and presented some results on the equivalence of free LCD codes. Meantime, Shi et al. \cite{Shi19} demonstrated that double circulant LCD codes exist over the ring $\mathbb{Z}_4$, and then enumerated such codes for length $2n$ over $\mathbb{Z}_4$. On the other hand, Prakash et al. \cite{ch6-Prakash22c} enumerated LCD double circulant codes of length $2n$ over a finite non-chain commutative ring. Moreover, Islam and Prakash \cite{Islam2022} utilized cyclic codes to construct LCD codes over a finite non-chain ring. Although LCD codes have been explored over different commutative rings, only a few works are available in the literature on LCD codes over noncommutative rings. This inspires us to explore LCD codes over a noncommutative ring. \par

Let $C$ and $C^{\prime}$ be two $\mathbb{F}_q$-linear codes of length $n$. Then, they are called monomial equivalent if there exists a monomial matrix $M_{n\times n }$ over $\mathbb{F}_q$ such that $C=C^{\prime}M=\{ \boldsymbol{x}M \mid \boldsymbol{x} \in C^{\prime} \}$. (Recall that a matrix is {\em monomial} if it contains exactly one nonzero entry per row and per column). In 2019, Araya and Harada \cite{Araya19} provided a comprehensive classification of binary and ternary LCD codes under monomial equivalence for lengths up to $13$ and $10$, respectively. Recently, in 2021, Shi et al. \cite{Shi21} introduced the LCD codes over the non-unital rings and investigated left LCD and left ACD codes over the ring $E_2$. We extend some of their results and initiate the study of MDS and AMDS LCD codes over $E_p$. Firstly, using the classification of binary and ternary LCD codes by Araya and Harada \cite{Araya19}, we have enumerated monomial inequivalent LCD codes over $E_2$ for lengths up to $13$ and over $E_3$ for lengths up to $10$. Then, inequivalent LCD codes over $E_2$ and $E_3$ with a fixed minimum distance for the said lengths are enumerated. Subsequently, we have associated MDS and AMDS LCD codes over $\mathbb{F}_p$ with MDS and AMDS codes over $E_p$ and classified MDS and AMDS LCD codes over $E_2$ and $E_3$ for lengths up to $6$. \par
On the other hand, if $C=C^\perp$, the linear code $C$ over $\mathbb{F}_q$ is known as a self-dual code. Since the 1970's research on self-dual codes has bloomed, largely driven by their many connections with lattices, designs, and, more recently, quantum error-correcting codes. In 2022,  Alahmadi et al. \cite{Alah22} studied quasi-self-dual codes over the noncommutative non-unital ring $E_2$ and classified them for short lengths. Then, in a follow-up paper, they considered three non-unital rings and studied various dualities of codes over these rings \cite{Alah23}. Moreover, they studied self-dual codes over the ring $E_p$ and classified them in short lengths \cite{Alah24}. Further, Kushwaha et al. \cite{Kushwaha24} studied quasi-self-dual codes over a noncommutative non-unital ring containing $9$ elements. Motivated by these works, we study two one-sided MDS and AMDS self-dual codes, such as MDS and AMDS left and right self-dual codes over $E_p$. Firstly, we study left self-dual codes over $E_p$ and classify MDS and AMDS left self-dual codes over $E_2$ and $E_3$ for lengths up to $12$. Then, we study right self-dual codes over $E_p$ and prove the non-existence of MDS right self-dual codes. Further, we show that AMDS right self-dual codes over $E_p$ exist only for length $2$. Finally, we study MDS and AMDS self-dual codes over $E_p$ and classify these codes over $E_2$ and $E_3$ for short lengths. As per our survey, this is the first attempt to study MDS and AMDS LCD and self-dual codes over a non-unital ring. \par
This work is structured as follows. Section $2$ contains basic definitions required for the subsequent sections. Section $3$ studies LCD codes over the ring $E_p$ and classifies MDS and AMDS LCD codes over $E_2$ and $E_3$ for some specific lengths. In Section $4$,  MDS and AMDS left self-dual codes are classified for lengths up to $12$. Moreover, we classify MDS and AMDS self-dual codes over $E_2$ and $E_3$ for short lengths. Section $5$ concludes the paper.

\section{Codes over the ring $E_p$}
For any prime $p$, let $E_p$ be a ring generated by two generators $r$ and $s$ under certain relations as follows:
$$ E_p= \langle r,s \mid pr =ps=0,~ r^2=r,~ s^2=s,~ rs=r,~ sr=s \rangle.$$
For an example, consider the ring $E_p$ generated by two matrices $r$ and $s$ over $\mathbb{F}_p$ where

$$
r=\begin{pmatrix}
0 ~& 0\\0~ & 1\\
\end{pmatrix},~~~
s=\begin{pmatrix}
0 ~& 1\\0~ & 1\\
\end{pmatrix}
.$$
Thus, the ring $ E_p=\{ ~ir+js~ |~ 0 \leq i,j <p ~\}$ contains $p^2$ elements and its characteristic is $p.$ We observe that the ring $E_p$ is noncommutative and has no multiplicative identity element. Also, $er=es=e$ for all $ e \in E_p.$ Let $t=r+(p-1)s$. Then, every element $e \in E_p$ has a $t$-adic decomposition as follows:
 $$e=ur+vt ~~ \text{where}~ u,v \in \mathbb{F}_p.$$
 Further,  $I=\{mt~ |~ 0 \leq m <p\}$ is the unique maximal ideal of the ring $E_p$.
%
Therefore, the ring $E_p$ is a local ring with residue class field $E_p/I\cong \mathbb{F}_p=\{0,1,\ldots, p-1\}$. Next, we can define a natural action of $\mathbb{F}_p$ on $E_p$ as $eu=ue$ for all $e \in E_p$ and $u \in \mathbb{F}_p$. Note that this action is distributive, i.e., for all $ e \in E_p$ and $u,v \in \mathbb{F}_p$, we have $e(u\oplus_{\mathbb{F}_p} v)=eu+ev=ue+ve$ where $\oplus_{\mathbb{F}_p}$ denotes the addition in $\mathbb{F}_p.$ Next, we write an arbitrary element of $E_p$ in $t$-adic decomposition form and define a  map $\alpha:E_p\rightarrow E_p/I =\mathbb{F}_p$ by
 $$\alpha(e)=\alpha(ur+vt)=u.$$
 This map is known as the map of reduction modulo $I$ and can be extended naturally from $E_p^n$ to  $\mathbb{F}_p^n$.

\begin{definition}[Linear code]\label{def1}
    A left $E_p$-submodule of $E_p^n$ refers to an $E_p$-linear code of length $n$.
\end{definition}

\begin{definition}[Generating set] Let $C$ be an $E_p$-linear code of length $n$  and  $X=\{ \boldsymbol{ x}_1,\boldsymbol{x}_2,\ldots , \boldsymbol{x}_k \} \subset C$. Then, the (left) $E_p$-span of $X$ is given by the set $$\langle X \rangle_{E_p} = \{ e_1\boldsymbol{x}_1+e_2\boldsymbol{x}_2+ \cdots + e_k\boldsymbol{x}_k ~|~ e_i \in E_p, ~  1 \leq i \leq k \},$$
 and the additive span of $X$ is given by the set
$$\langle X \rangle_{\mathbb{F}_p} = \{ u_1\boldsymbol{x}_1+u_2\boldsymbol{x}_2+ \cdots + u_k\boldsymbol{x}_k ~|~ u_i \in \mathbb{F}_p, ~ 1 \leq i \leq k  \}.$$
Note that $\langle X \rangle_{E_p}$ does not always contain $\langle X \rangle_{\mathbb{F}_p}$ since the ring $E_p$ does not have the unity element. A subset $X=\{ \boldsymbol{x}_1,\boldsymbol{x}_2,\ldots , \boldsymbol{x}_k \}$ of $C$ is called a generating set for the code $C$ if
    $$\langle X \rangle_{E_p} \cup \langle X \rangle_{\mathbb{F}_p}=C.$$

\end{definition}

\begin{definition}[Generator matrix] Let $C$ be an $E_p$-linear code  of length $n$  and $X=\{\boldsymbol{x}_1,\boldsymbol{x}_2,\ldots, \boldsymbol{x}_k \}\subset C$ be its generating set. Then, a generator matrix $G_{E_p}$ of $C$ is a $k \times n$ matrix  whose rows are $\boldsymbol{x}_1,\boldsymbol{x}_2,\ldots, \boldsymbol{x}_k$ and $\langle G \rangle_{E_p}=\langle X \rangle_{E_p} \cup \langle X \rangle_{\mathbb{F}_p}$.

\end{definition}

\begin{definition}[Minimum (Hamming) weight and distance of a linear code]
The number of non-zero positions in a codeword of a linear code $C$ is called the weight of that codeword, and the minimum of all the non-zero weights is defined as the minimum weight of the code $C$. Next, the number of positions at which two distinct codewords of a linear code $C$ differ is called the distance between them. Then, the minimum of all the distances between two distinct codewords of $C$ is called the minimum distance of the code $C$. Usually, $wt(C)$ and $d(C)$ denote the minimum weight and the distance of a linear code $C$, respectively.

\end{definition}


 Now, we define two linear codes over $\mathbb{F}_p$ associated with a linear code $C$ of length $n$ over $E_p$.
\begin{enumerate}
\item[(i)] \textbf{Residue code:} The residue code of the code $C$ is defined by
    $$ Res(C)=\{\alpha(\boldsymbol{x})~ |~ \boldsymbol{x} \in C \}.$$
    \item[(ii)] \textbf{Torsion code:} The torsion code of the code $C$ is given by
    $$ Tor(C)=\{ \boldsymbol{v} \in \mathbb{F}_{p}^{n} ~ |~ \boldsymbol{v}t \in C\},$$ \\
    where $t=r+(p-1)s.$
    \end{enumerate}
%
    \begin{remark}
        Throughout the paper, $r$ and $s$ represent the generators of the ring $E_p$ satisfying the relations $pr=ps=0,~r^2=r,~ s^2=s,~ rs=r,~ sr=s $, and $t=r+(p-1)s$. In addition, $m_1$ and $m_1+m_2$ denote the dimensions of $Res(C)$ and $Tor(C)$, respectively.
     \end{remark}

    Next, for any $\boldsymbol{x}=(x_1, x_2, \ldots, x_n),~\boldsymbol{y}=(y_1,y_2,\ldots,y_n) \in E_p^n$, define an inner product on $E_p^n$ by
    $$ \langle \boldsymbol{x}, \boldsymbol{y}  \rangle=\sum_{j=1}^{n}x_j y_j.$$
    Under this inner product, the following two duals of an $E_p$-linear code $C$ of length $n$ are defined.
    \begin{enumerate}
     \item[(i)] \textbf{Left dual:} The left dual of the code $C$ is defined as
    $$C^{\perp_L}= \{ \boldsymbol{z} \in E_p^n~ |~\langle \boldsymbol{z}, \boldsymbol{w}  \rangle= 0,\forall ~\boldsymbol{w}\in C\}.$$

    \item[(ii)] \textbf{Right dual:} The right dual of the code $C$ is  defined as
        $$C^{\perp_R}= \{ \boldsymbol{z }\in E_p^n~ |~\langle \boldsymbol{w}, \boldsymbol{z}  \rangle= 0,\forall ~\boldsymbol{w} \in C\}.$$
    \end{enumerate}

\section{LCD codes over the ring $E_p$}
This Section focuses on LCD codes over the ring $E_p$ and enumerates monomial inequivalent LCD codes over $E_2$ and $E_3$ for lengths up to $13$ and $10$, respectively. Also, the number of monomial inequivalent LCD codes with a given code length and minimum distance is calculated over $E_2$ and $E_3$. Moreover, we study MDS and AMDS LCD  codes over $E_p$ and classify these codes over $E_2$ and $E_3$ under monomial equivalence for lengths up to $6$.

\begin{definition}[Left and right nice codes]
   An $E_p$-linear code $C$ of length $n$ is said to be left nice if it satisfies  $|C|\cdot|C^{\perp_L}|=|E_p|^n$, and is called right nice if it satisfies  $|C|\cdot|C^{\perp_R}|=|E_p|^n$.
\end{definition}



\begin{definition}[Left and right LCD codes] If a left (resp. right) nice code $C$ over $E_p$ satisfies  $C\cap C^{\perp_L}=\{\boldsymbol{0}\}$ (resp. $C\cap C^{\perp_R}=\{\boldsymbol{0}\}$), it is called left (resp. right) LCD code.
\end{definition}

Note that, for an $E_p$-linear code $C$, $C\oplus C^{\perp_L}=E_p^n$ (resp. $C\oplus C^{\perp_R}=E_p^n$)  if and only if it is a left (resp. right) LCD code over $E_p$.

\begin{definition}[Free code] An $E_p$-linear code $C$ is called free if it can be expressed as a finite direct sum of $E_p$ ($E_p$ as a left $E_p$-module), i.e., $C=E_p\oplus E_p\oplus \cdots \oplus E_p$ where $E_p=\langle x_i\rangle_{E_p}$ for some $x_i \in E_p.$
\end{definition}

Following \cite{Alah22}, we observe that an $E_p$-linear code $C$ is free if and only if $Res(C)$ and $Tor(C)$ are equal, or equivalently, $m_2=0$. \par
The following few results for linear codes over $E_p$ are simple extensions of some results on linear codes over $E_2$ from \cite{Shi21}. We skip their proofs as it simply replaces binary codes with codes over $\mathbb{F}_p$ and binary matrices with matrices over $\mathbb{F}_p$.

\begin{remark}

 There exists no non-zero right LCD code over $E_p$.
 \end{remark}

Therefore, we consider the left LCD codes as LCD codes in our further investigations.

\begin{theorem}\label{thm2}
 Let $C$ be a free linear code over $E_p$ with generator matrix $G_{E_p}$. Then
 \begin{enumerate}
     \item[(i)] $C=\langle rG \rangle_{E_p}$,
     \item[(ii)] $C^{\perp_L}=\langle rH \rangle_{E_p}$,
 \end{enumerate}
 where $G$ is a generator matrix of $Res(C)$ and $H$ is a parity-check matrix of $Res(C)$.
\end{theorem}

\begin{theorem}\label{thm1e}
    An LCD code over $E_p$ is always free.
\end{theorem}

Authors in \cite{Shi21} have demonstrated a method for constructing LCD codes over $E_2$ using LCD codes over $\mathbb{F}_2$. In the next result, we extend that method for constructing LCD codes over $E_p$ by using LCD codes over $\mathbb{F}_p$. We omit the proof as it follows the same procedure given in Proposition 2 in \cite{Shi21}.

\begin{theorem}\label{thm3}
    Let the matrix $G_{k \times n}$ generates an LCD code $D$ over $\mathbb{F}_p$. Then the $E_p$-span of $rG$, $\langle rG\rangle _{E_p}$ is an LCD code over $E_p$.
\end{theorem}

\begin{theorem}\label{thm1a}\cite{Shi21}
     A free $E_2$-LCD code $C$ is generated by the matrix $rG_2$ where $G_2$ generates a binary LCD code $D$.
\end{theorem}

In the next result, we extend Theorem \ref{thm1a} for LCD codes over $E_p$ and remove the word ``free", as LCD codes are always free by Theorem \ref{thm1e}. We omit the proof as it follows the same procedure as in the proof of Theorem \ref{thm1a}.

\begin{theorem}\label{thm4}
     An LCD code $C$ over $E_p$ has the generator matrix of the form $rG$ for some generator matrix $G$ of an $\mathbb{F}_p$-LCD code $D$.
\end{theorem}

The next result gives a relationship between an LCD code over $E_p$ and an LCD code over $\mathbb{F}_p$.
\begin{theorem}\label{thm5e}
     An $E_p$-linear code $C$ is LCD if and only if it is generated by a matrix of the form $rG$ where $G$ generates an $\mathbb{F}_p$-LCD code $D$.
\end{theorem}
\begin{proof}
   By Theorem \ref{thm4}, an LCD code $C$ over $E_p$ is generated by a matrix of the form $rG$ where $G$ generates an $\mathbb{F}_p$-LCD code $D$. \par

 Conversely, suppose that $C$ is an $E_p$-linear code and $rG$ is its generator matrix where $G$ generates an $\mathbb{F}_p$-LCD code $D$. By Theorem \ref{thm3}, the $E_p$-span of $rG$, $\langle rG\rangle _{E_p}$ is an LCD code over $E_p$. It is easy to see that the $E_p$-span of $rG$ always contains the additive span of $rG$ by taking scalars as $0$ or $r$. Therefore, the code $C$ generated by the matrix $rG$ is an LCD code over $E_p$.
\end{proof}

\begin{definition}[Monomial equivalent codes]
    Two $E_p$-linear codes $C$ and $C'$ of the same length $n$ are called monomial equivalent if there exists a monomial matrix $M_{n\times n}$  over $\mathbb{F}_p$ (a matrix consisting only one non-zero entry from $\mathbb{F}_p$ in each row and each column) such that $C=C'M=\{xM \mid x\in C'\}$.
\end{definition}

 The following result is crucial for our study, as it gives a condition under which two free linear codes over $E_p$  become monomial equivalent.

\begin{theorem}\label{thm4E}
    Let $B$ and $C$ be two free $E_p$-linear codes of the same length $n$. Then, $B$ and $C$ are monomial equivalent if and only if their residue codes $Res(B)$ and $Res(C)$ are monomial equivalent.
\end{theorem}
\begin{proof}
  Let $B$ and $C$ be two free $E_p$-linear codes of the same length $n$. Suppose that these two codes are monomial equivalent. Then, there exists a monomial matrix $M$ over $\mathbb{F}_p$  such that $B=CM$. Since $\alpha(B)=\alpha(CM)=\alpha(C)M$, $Res(B)=Res(C)M$. Therefore, the residue codes $Res(B)$ and $Res(C)$ are monomial equivalent. \par
   Conversely, suppose that the residue codes $Res(B)$ and $Res(C)$ are monomial equivalent. By Theorem \ref{thm2}, the free $E_p$-linear codes $B$ and $C$ will have generator matrices of the form $rG_1$ and $rG_2$, respectively, where $G_1$ and $G_2$ are generator matrices of their respective residue codes. Since $Res(B)$ and $Res(C)$ are monomial equivalent, there exists a monomial matrix $M$ over $\mathbb{F}_p$ such that $Res(B)=Res(C)M$. Without loss of generality, we may assume that $G_1=G_2M.$ Since $G_1=G_2M$, $rG_1=rG_2M$. This implies that $\langle rG_1 \rangle_{E_p}=\langle rG_2M \rangle _{E_p}= \langle rG_2 \rangle_{E_p} M$. Therefore, $B=CM$. Thus, the free $E_p$-linear codes $B$ and $C$ are monomial equivalent.

\end{proof}

Next, we give an example which shows that Theorem \ref{thm4E} may not be valid for non-free linear codes over $E_p$.
\begin{example}
 Let $$B=\{ 000,rr0,ss0,tt0,0t0,rs0,sr0,t00\}$$ and $$C=\{000,r0r,s0s,t0t,0t0,rtr,sts,ttt\}$$  be two $E_2$-linear codes. Then, $Res(B)=\{000, 110\}$, $Res(C)=\{000, 101\}$, $Tor(B)=\{000,110,010,100\}$ and $Tor(C)=\{000,101,010,111\}$. We see that the $E_2$-linear codes $B$ and $C$ are not free, as their respective torsion and residue codes are unequal. Next, a short calculation shows that $Res(B)$ and $Res(C)$ are monomial equivalent over $\mathbb{F}_2$. On the other hand, $B$ has $2$ codewords of weight $1$, and $C$ has only one codeword of weight $1$. Therefore, the $E_2$-linear codes $B$ and $C$ are not monomial equivalent.
\end{example}

The next result enumerates monomial inequivalent LCD codes over $E_p$ for a fixed code length.
\begin{theorem}\label{thm7a}
    The number of monomial inequivalent $E_p$-LCD codes and the number of monomial inequivalent  $\mathbb{F}_p$-LCD codes of the same lengths are always equal.
\end{theorem}
\begin{proof}
   By Theorem \ref{thm1e}, an LCD code $C$ over $E_p$ is always free. Moreover, by Theorem \ref{thm4E}, two free $E_p$-codes $B$ and $C$ are monomial equivalent if and only if their residue codes $Res(B)$ and $Res(C)$ are monomial equivalent. Then, the proof follows directly from Theorem \ref{thm5e}.
\end{proof}

Our next purpose is to enumerate monomial inequivalent LCD codes over $E_2$ and $E_3$ for a fixed code length. We use Propositions $4$, $5$ and Tables $4$, $6$ from \cite{Araya19} together with Theorem \ref{thm7a} to enumerate monomial inequivalent LCD codes of a fixed length $n$ over $E_2$ and $E_3$ and list them in Tables \ref{Tab1} and  \ref{Tab2}, respectively. In the tables, the number of monomial inequivalent LCD codes of length $n$ is denoted by $N$. These numbers have been calculated over $E_2$ up to length $13$ and over $E_3$ up to length $10.$

\begin{longtable}{|c|c|c|c|c|c|c|c|c|c|c|c|c|c|c|}

\caption{\label{Tab1} No. of monomial inequivalent LCD codes over $E_2$.}\\
\hline

 $n$ & $1$ & $2$ & $3$ & $4$ & $5$ & $6$ & $7$ & $8$ & $9$ & $10$ & $11$ & $12$ & $13$ \\

 \hline
 $N$ & $2$ & $3$ & $6$ & $10$ & $18$ & $34$& $66$ & $138$ & $312$ & $790$ & $2234$ & $7534$ & $30620$ \\

 \hline

\end{longtable}

\begin{longtable}{|c|c|c|c|c|c|c|c|c|c|c|c|}
\caption{\label{Tab2} No. of monomial inequivalent LCD codes over $E_3$.} \\

\hline
 $n$ & $1$ & $2$ & $3$ & $4$ & $5$ & $6$ & $7$ & $8$ & $9$ & $10$ \\

\hline
 $N$ & $2$ & $4$ & $6$ & $12$ & $24$ & $49$& $116$ & $331$ & $1136$ & $5590$   \\

 \hline

\end{longtable}

In the next result, we provide a lower bound for the number of monomial inequivalent LCD codes over $E_3$ for any arbitrary length $n$.
\begin{theorem}
    Let $\phi(n,m)$ be the number of ternary LCD codes of length $n$ and dimension $m$. Then there exist at least $ \sum_{m=0}^{m=n}\big\lceil\frac{\phi(n,m)}{{2^{n-1}n!}} \big\rceil$  monomial inequivalent $E_3$-LCD codes of length $n$.
\end{theorem}
\begin{proof}
   From \cite{Araya19}, the automorphism group of a ternary LCD code has at least $2$ elements. Hence, there are at least $\big\lceil\frac{\phi(n,m)}{2^{n-1}n!}\big\rceil$ monomial inequivalent $m$-dimensional ternary LCD codes of length $n$. The rest part of the proof holds immediately from Theorem \ref{thm7a}.
\end{proof}

The next two results of \cite{Alah24} are essential in our further studies.
\begin{theorem}\label{thm5}\cite{Alah24}
 For an $E_p$-linear code $C$, $rRes(C) \subseteq C$.
\end{theorem}

\begin{theorem}\label{thm7}\cite{Alah24}
 For an $E_p$-linear code $C$, $C=rRes(C) \oplus tTor(C)$.
\end{theorem}

The next result relates the minimum distance of a free $E_p$-linear code to the minimum distance of its residue code.
\begin{theorem}\label{thm 5E}
   Let $C$ be a non-zero free $E_p$-linear code. Then the minimum distances of $C$ and  $Res(C)$ are equal.
\end{theorem}
\begin{proof}
Let $d(C)=d_C$ and $d(Res(C))=d_R$. Since $d(Res(C))=d_R$, there must exist some non-zero $\boldsymbol{z_1} \in Res(C)$ with $wt(\boldsymbol{z_1})=d_R$. Since $rRes(C) \subseteq C$ by  Theorem \ref{thm5}, $r\boldsymbol{z_1}\in C$. Then, $wt(r\boldsymbol{z_1})=wt(\boldsymbol{z_1})=d_R$ implies that  $d_C \leq d_R$. Further, let $\boldsymbol{z} \in C$ with $wt(\boldsymbol{z})=d_C$. This shows that $\boldsymbol{z}=r\boldsymbol{z_1}+t\boldsymbol{z_2}$ for some $\boldsymbol{z_1}, \boldsymbol{z_2}  \in Res(C)$, by Theorem \ref{thm7}. Then, the following cases arise:
\begin{enumerate}

    \item When $\boldsymbol{z_1}\neq \boldsymbol{0}$ and $\boldsymbol{z_2}= \boldsymbol{0}$,  $wt(\boldsymbol{z})=wt(r\boldsymbol{z_1})=wt(\boldsymbol{z_1})$.
    \item When $\boldsymbol{z_1}=\boldsymbol{0}$ and $\boldsymbol{z_2}\neq \boldsymbol{0}$,  $wt(\boldsymbol{z})=wt(t\boldsymbol{z_2})=wt(\boldsymbol{z_2})$.
    \item When $\boldsymbol{z_1}, \boldsymbol{z_2}\neq \boldsymbol{0}$,  $wt(\boldsymbol{z}) \geq wt(r\boldsymbol{z})=wt(r\boldsymbol{z_1})=wt(\boldsymbol{z_1})$.
\end{enumerate}
From the above three cases, we conclude that $d_C=wt(\boldsymbol{z})\geq d_R$, as $\boldsymbol{z_1}, \boldsymbol{z_2} \in Res(C)$ and $wt(\boldsymbol{z_1}), wt(\boldsymbol{z_2})\geq d_R$. Thus, $d_C=d_R$.

\end{proof}

Next, we give an example which shows that Theorem \ref{thm 5E} may not be true for non-free $E_p$-codes.
\begin{example}
 Let $C=\{ 000,rr0,ss0,tt0,0t0,rs0,sr0,t00\}$ be an $E_2$-linear code. Then, $Tor(C)=\{000,110,010,100\}$ and  $Res(C)=\{000,110\}$. The $E_2$-linear code $C$ is not free, as its torsion and residue codes are unequal. Further, $d(Res(C))=2$ and $d(C)=1$  imply that $ d(Res(C)) \neq d(C)$.
\end{example}

The next result enumerates monomial inequivalent $E_p$-LCD codes with a fixed code length and minimum distance.
\begin{theorem}\label{thm11}
    The number of monomial inequivalent LCD codes of length $n$ over $E_p$ with minimum distance $d$ is equal to the number of monomial inequivalent LCD codes of length $n$ over $\mathbb{F}_p$ with minimum distance $d$.
\end{theorem}
\begin{proof}
By Theorem \ref{thm4}, we know that the generator matrix of an $E_p$-LCD code $C$ is of the form $rG$ where $G$ generates an $\mathbb{F}_p$-LCD code $D$. In addition, by Theorem \ref{thm 5E}, the minimum distances of the $E_p$-LCD code $C$ and $Res(C)$ are always equal. Moreover, $Res(C)=D$ is an LCD code over $\mathbb{F}_p$. Then the proof holds immediately from Theorem \ref{thm7a}.
\end{proof}

Our next purpose is to enumerate monomial inequivalent LCD codes over $E_2$ and $E_3$ for a fixed length and minimum distance. We use Propositions $4$, $5$ and Tables $4$, $6$ from \cite{Araya19} along with Theorem \ref{thm11} to calculate the numbers of monomial inequivalent LCD codes of length $n$ over $E_2$ and $E_3$ with a minimum distance $d$ and list them in Tables \ref{Tab3} and \ref{Tab4}, respectively. We denote by $N_d$ the number of monomial inequivalent LCD codes of length $n$ and minimum distance $d$. These numbers are given over $E_2$ up to length $13$ and over $E_3$ up to length $10.$
\begin{table}[ht]
{
\centering
\caption{\label{Tab3} No. of monomial inequivalent LCD codes over $E_2$ with a fixed minimum distance.}
\begin{tabular}{|c|ccccccccccccc|c|}
\hline
 $n$ & $N_1$ & $N_2$ & $N_3$ & $N_4$ & $N_5$ & $N_6$ & $N_7$ & $N_8$ & $N_9$ & $N_{10}$ & $N_{11}$ & $N_{12}$ & $N_{13}$ \\

\hline
 $1$ & $1$ &&&&&&&&&&&& \\

  $2$ & $2$ &-&&&&&&&&&&&\\

   $3$ & $3$ & $1$& $1$ &&&&&&&&&& \\

   $4$ & $6$ & $2$ & $1$ &-&&&&&&&&& \\

   $5$ & $10$ & $5$ & $1$ &- & $1$ &&&&&&&&\\

   $6$ & $18$ & $11$ & $3$ & -& $1$ &-&&&&&&& \\

   $7$ & $34$ & $24$ & $4$ & $1$ & $1$ &- & $1$ &&&&&&\\
   $8$ & $66$ & $57$ & $9$ & $2$ & $2$ &- & $1$ &- &&&&& \\
 $9$ & $138$ & $140$ & $23$ & $5$ & $2$ & $1$ & $1$ &-  & $1$ &&&& \\
 $10$ & $312$ & $393$ & $61$ & $13$ & $6$ & $2$ & $1$ &-  & $1$ &- &&& \\
 $11$ & $790$ & $1199$ & $185$ & $41$ & $11$ & $4$ & $1$ &-  & $1$ &- & $1$ && \\
 $12$ & $2234$ & $4381$ & $726$ & $152$ & $28$ & $7$ & $3$ &-  & $1$ &- & $1$ &- &\\
 $13$ & $7534$ & $18717$ & $3558$ & $708$ & $79$ & $16$ & $3$ & $1$ & $1$ &-  & $1$ &- & $1$  \\
\hline

\end{tabular}

}
\end{table}

\begin{table}[ht]
{
\centering
\caption{\label{Tab4} No. of monomial inequivalent LCD codes over $E_3$ with a fixed minimum distance.}
\begin{tabular}{|c|cccccccccc|c|}
\hline
 $n$ & $N_1$ & $N_2$ & $N_3$ & $N_4$ & $N_5$ & $N_6$ & $N_7$ & $N_8$ & $N_9$ & $N_{10}$  \\

\hline
 $1$ & $1$ &&&&&&&&& \\

  $2$ & $2$ & $1$ &&&&&&&&\\

   $3$ & $4$ & $1$& - &&&&&&&\\

   $4$ & $6$ & $4$ & - &$1$ &&&&&& \\

   $5$ & $12$ & $8$ & $1$ &$1$ & $1$ &&&&&\\

   $6$ & $24$ & $17$ & $4$ & $2$& $1$ &-&&&& \\

   $7$ & $49$ & $47$ & $12$ & $5$ & $1$ &- & $1$ &&&\\
   $8$ & $116$ & $149$ & $43$ & $17$ & $3$ &- & $1$ &$1$ && \\
 $9$ & $331$ & $539$ & $191$ & $63$ & $8$ & $1$ & $1$   & $1$ &- & \\
 $10$ & $1136$ & $2732$ & $1266$ & $410$ & $37$ &  $4$ & $2$ & $1$ &-  & $1$  \\

\hline

\end{tabular}

}
\end{table}

    \subsection{MDS and AMDS LCD codes over $E_p$}
For an $(n,M,d)$-code, the Singleton bound is already defined in \cite{Singleton}. Here, we define the same bound over the ring $E_p$.
\begin{definition}[Singleton bound]
    For an $(n,M,d)$-code $C$ over $E_p$, the Singleton bound states that $M \leq |E_p|^{n-d+1}$.
\end{definition}
%
\begin{definition}[MDS and AMDS codes]
    For a linear $(n,M,d)$-code $C$ over $E_p$, if $M = |E_p|^{n-d+1}$ (resp. $M = |E_p|^{n-d}$), it is known as an MDS (resp. AMDS) code.
    \end{definition}

The next result is crucial for our study, as it characterizes free MDS and AMDS codes over $E_p$.
\begin{theorem}\label{thm13}
    A free linear code $C$ of length $n$ over $E_p$ with minimum distance $d$ is an MDS (resp. AMDS) code over $E_p$ if and only if its residue code is an MDS (resp. AMDS) code over $\mathbb{F}_p.$
\end{theorem}
\begin{proof}
  The torsion and residue codes of a free $E_p$-linear code $C$ are always equal. Hence, $$|C|=|Res(C)||Tor(C)=|Res(C)|^2=(p^2)^{m_1},$$
  and $$|E_p|^{n-d+1}=(p^2)^{n-d+1}.$$
  These imply that  $|C| = |E_p|^{n-d+1}$ if and only if $m_1=n-d+1$. Moreover, by Theorem \ref{thm 5E}, $d(C)=d(Res(C))$. Therefore, the free $E_p$-linear code $C$ is MDS if and only if its residue code is MDS. Similar arguments work for the AMDS part.
\end{proof}

The following corollary is very important, as it provides a necessary and sufficient condition for an LCD code $C$ over $E_p$ to be MDS or AMDS.

\begin{corollary}\label{cor1}
   An $E_p$-LCD code $C$ is MDS (resp. AMDS)  if and only if its residue code is an LCD MDS (resp. AMDS) code over $\mathbb{F}_p$.
\end{corollary}
\begin{proof}
 By Theorem \ref{thm1e}, an $E_p$-LCD code  $C$  is always free. Then the proof holds immediately from Theorem \ref{thm13}.
\end{proof}

The following examples show that Theorem \ref{thm13} may not hold for non-free linear codes over $E_p$.
\begin{example}
    Let $C$ be an $E_2$-linear code with generator matrix
    $$G_{E_2}=\begin{pmatrix}
        r &~ r & ~r &~ r\\
        t &~ 0 & ~0 & ~t\\
        0 & ~t & ~0 & ~t\\
        0 & ~0 &~ t & ~t
    \end{pmatrix}.$$
    Then, its residue and torsion codes have respective generator matrices
    $$G_1=\begin{pmatrix}
        1 &~ 1&~ 1&~ 1
    \end{pmatrix} ~ \text{and}~ G_2=\begin{pmatrix}
        1 &~ 0 &~ 0 &~ 1 \\
        0 &~ 1 &~ 0 &~ 1 \\
        0 &~ 0 &~ 1 &~ 1
    \end{pmatrix}.$$
    This shows that $C$ is not free. We can easily see that the residue code is a $(4,2,4)$ binary code. Therefore, the residue code is an MDS code. But the $E_2$-linear code $C$ is a $(4,16,2)$-code which is not MDS as $|C|=2^4$ and $|E_2|^{n-d+1}=2^6$.
\end{example}
\begin{example}
    Let $C=\{000,r0r,s0s,t0t,0t0,rtr,sts,ttt\}$ be an $E_2$-linear code of length $3$. Then,  $Res(C)=\{000, 101\}$ and $Tor(C)=\{000,101,010,111\}$. This linear code $C$ over $E_2$ is not free as its torsion and residue codes are not equal. Further, we see that $Res(C)$ is a $(3,2,2)$-code which is an AMDS code. But $C$ is a $(3,8,1)$-code over $E_2$ which is not AMDS, as $|C|=2^3$ and $|E_2|^{n-d}=2^4$.
\end{example}

The next result enumerates monomial inequivalent MDS and AMDS LCD codes over $E_p$.
\begin{theorem}
    The number of monomial inequivalent MDS (resp. AMDS) LCD codes over $E_p$ and the number of monomial inequivalent MDS (resp. AMDS) LCD codes over $\mathbb{F}_p$ having the same lengths and minimum distances are equal.
\end{theorem}
\begin{proof}
   The proof follows immediately from Theorems \ref{thm11} and \ref{thm13}.
\end{proof}

Our next purpose is to classify MDS and AMDS LCD codes over $E_2$ and $E_3$ for some specific lengths. We rely on \cite{Araya19} for inequivalent binary and ternary LCD codes and use MAGMA  \cite{Magma} to find MDS and AMDS binary and ternary LCD codes. Then, we use Theorems \ref{thm2}, \ref{thm1e}, \ref{thm4E}, \ref{thm 5E}, and Corollary \ref{cor1} to classify monomial inequivalent MDS and AMDS LCD codes over $E_2$ and $E_3$ for lengths up to $6$ and list them in Tables \ref{Tab5} and \ref{Tab6}, respectively.

\begin{longtable}{|c|p{6cm}|c|c|}

 \caption{\label{Tab5} Monomial inequivalent MDS and AMDS LCD codes over $E_2$.}\\

\hline
 $n$ &~~~~~~~~ Generator Matrix & Minimum  & Remark \\ && Distance &  \\

 \hline

 $1$ & \hspace{2.3cm}$\begin{pmatrix}
     r
 \end{pmatrix}$ & $1$ & MDS  \\[5pt]
\hline

  & \hspace{2.2cm}$\begin{pmatrix}
     r& 0
 \end{pmatrix}$  & 1 & AMDS\\ [5pt]
 $2$ &\hspace{2.25cm}$rI_2$ & 1 & MDS \\ [5pt]
 \hline

 & \hspace{2cm} $\begin{pmatrix}
     r & r & r
 \end{pmatrix}$ & 3 & MDS \\ [5pt]
 $3$ & \hspace{2cm} $\begin{pmatrix}
     r & r & 0 \\
     r & 0 & r
 \end{pmatrix}$ & 2 & MDS \\[10pt]
  & \hspace{2cm} $\begin{pmatrix}
     0 & r & 0 \\
     0 & 0 & r
 \end{pmatrix}$ & 1 & AMDS \\[5pt]
 & \hspace{2.3cm} $rI_3$  & 1 & MDS \\ [5pt]
 \hline

 & \hspace{2cm}$\begin{pmatrix}
     r & r & r & 0
 \end{pmatrix}$ & $3$ & AMDS \\ [5pt]

 & \hspace{2cm}$\begin{pmatrix}
     r & 0 & r & r \\
     0 & r & r & r
 \end{pmatrix}$ & 2 & AMDS \\ [9pt]

& \hspace{2cm}$\begin{pmatrix}
      r & 0 & 0 & r \\
     0 & r & 0 & r
 \end{pmatrix}$ & $2$ & AMDS \\ [9pt]
 $4$ &\hspace{1.82cm} $\begin{pmatrix}
     0 & r & 0 & 0 \\
     0 & 0 & r & 0  \\
     0 & 0 & 0 & r
 \end{pmatrix}$ & 1 & AMDS\\ [7pt]
&\hspace{1.82cm} $\begin{pmatrix}
     r & r & 0 & 0 \\
     r & 0 & r & 0  \\
     0 & 0 & 0 & r
 \end{pmatrix}$ & 1 & AMDS  \\ [15pt]
& \hspace{2.4cm} $rI_4$ & 1 & MDS \\ [5pt]
\hline

&\hspace{2cm} $\begin{pmatrix}
    r & r & r & r &r
\end{pmatrix}$ & 5 & MDS \\ [5pt]

& \hspace{2cm}$\begin{pmatrix}
    0 & r & r & 0 & 0 \\
    0 & r & 0 & r & 0 \\
    r & r & 0 & 0 & r
\end{pmatrix}$ & 2 & AMDS \\ [15pt]

$5$ &\hspace{1.85cm} $\begin{pmatrix}
    0 & r & 0 & 0 & 0 \\
    0 & 0 & r & 0 & 0 \\
    0 & 0 & 0 & r & 0 \\
    0 & 0 & 0 & 0 & r
\end{pmatrix}$ & 1 & AMDS \\ [20pt]

& \hspace{2cm}$\begin{pmatrix}
    r & r & 0 & 0 & 0 \\
    r & 0 & r & 0 & 0 \\
    r & 0 & 0 & r & 0 \\
    r & 0 & 0 & 0 & r
\end{pmatrix}$ & 2 & MDS \\ [20pt]
& \hspace{2cm}$\begin{pmatrix}
    r & r & 0 & 0 & 0 \\
    r & 0 & r & 0 & 0 \\
    0 & 0 & 0 & r & 0 \\
    0 & 0 & 0 & 0 & r
\end{pmatrix}$ & 1 & AMDS \\ [20pt]

& \hspace{2.6cm} $rI_5$ & 1 & MDS \\ [5pt]
 \hline

 & \hspace{2.1cm}$\begin{pmatrix}
     r & 0 & r & r & r & r
 \end{pmatrix}$ & 5 & AMDS \\ [5pt]

 & \hspace{2cm}$\begin{pmatrix}
    0 & r & r & 0 & 0 & 0 \\
    0 & r & 0 & r & 0 & 0 \\
    0 & r & 0 & 0 & r & 0 \\
    r & r & 0 & 0 & 0 & r
\end{pmatrix}$ & 2 & AMDS \\ [20pt]
& \hspace{2cm}$\begin{pmatrix}
    r & r & r & 0 & 0 & 0 \\
    r & 0 & 0 & r & 0 & 0 \\
    0 & r & 0 & 0 & r & 0 \\
    0 & r & 0 & 0 & 0 & r
\end{pmatrix}$ & 2 & AMDS\\ [20pt]

& \hspace{2cm}$\begin{pmatrix}
    r & 0 & r & 0 & 0 & 0 \\
    r & 0 & 0 & r & 0 & 0 \\
    0 & r & 0 & 0 & r & 0 \\
    0 & r & 0 & 0 & 0 & r
\end{pmatrix}$ & 2 & AMDS \\ [20pt]
 $6$  & \hspace{2cm}$\begin{pmatrix}
    0 & r & r & 0 & 0 & 0 \\
    0 & r & 0 & r & 0 & 0 \\
    0 & r & 0 & 0 & r & 0 \\
    0 & r & 0 & 0 & 0 & r
\end{pmatrix}$ & 2 & AMDS  \\ [20pt]
& \hspace{2cm}$\begin{pmatrix}
    0 & r & 0 & 0 & 0 & 0 \\
    r & 0 & r & 0 & 0 & 0 \\
    r & 0 & 0 & r & 0 & 0 \\
    r & 0 & 0 & 0 & r & 0 \\
    r  & 0 & 0 & 0 & 0 & r
\end{pmatrix}$ & 1 & AMDS \\ [25pt]
&\hspace{1.85cm} $\begin{pmatrix}
    0 & r & 0 & 0 & 0 & 0 \\
    0 & 0 & r & 0 & 0 & 0 \\
    0 & 0 & 0 & r & 0 & 0 \\
    0 & 0 & 0 & 0 & r & 0 \\
    0  & 0 & 0 & 0 & 0 & r
\end{pmatrix}$ & 1 & AMDS \\ [25pt]
& \hspace{2cm}$\begin{pmatrix}
    0 & r & 0 & 0 & 0 & 0 \\
    0 & 0 & r & 0 & 0 & 0 \\
    0 & 0 & 0 & r & 0 & 0 \\
    r & 0 & 0 & 0 & r & 0 \\
    r  & 0 & 0 & 0 & 0 & r
\end{pmatrix}$ & 1 & AMDS \\ [5pt]
& \hspace{2.7cm} $rI_6$ & 1 & MDS \\ [3pt]

\hline

\end{longtable}

\begin{longtable}{|c|p{6cm}|c|c|}
\caption{\label{Tab6} Monomial inequivalent MDS and AMDS LCD codes over $E_3$.}\\
\hline
 $n$ &~~~~~~~~ Generator Matrix & Minimum  & Remark \\ && Distance &  \\

 \hline

 $1$ & \hspace{2.25cm}$\begin{pmatrix}
     r
 \end{pmatrix}$ & $1$ & MDS  \\[5pt]
\hline

  & \hspace{2.15cm}$\begin{pmatrix}
     r& r
 \end{pmatrix}$  & 2 & MDS\\ [5pt]
 $2$& \hspace{2.15cm}$\begin{pmatrix}
      r & 0
 \end{pmatrix}$ & 1 & AMDS \\
 &\hspace{2.25cm}$rI_2$ & 1 & MDS \\ [5pt]
 \hline

 & \hspace{2.1cm} $\begin{pmatrix}
     r & r & 0
 \end{pmatrix}$ & 2 & AMDS \\ [5pt]
 $3$ & \hspace{2cm} $\begin{pmatrix}
     2r & r & 0
 \end{pmatrix}$ & 1 & AMDS \\ [5pt]
 & \hspace{2cm} $\begin{pmatrix}
     0 & r & 0 \\
     0 & 0 & r
 \end{pmatrix}$ & 1 & AMDS \\[5pt]
 & \hspace{2.3cm} $rI_3$  & 1 & MDS \\ [5pt]
 \hline

 & \hspace{2cm}$\begin{pmatrix}
     r & 2r & 2r & r
 \end{pmatrix}$ & $4$ & MDS \\ [5pt]

 & \hspace{2cm}$\begin{pmatrix}
     r & 0 & r & 0 \\
     0 & r & 2r & r
 \end{pmatrix}$ & 2 & AMDS  \\ [9pt]
 & \hspace{2cm}$\begin{pmatrix}
      r & 0 & r & 0 \\
     0 & r & 0 & r
 \end{pmatrix}$ & $2$ & AMDS \\ [9pt]
 $4$ & \hspace{2cm} $\begin{pmatrix}
     0 & r & 0 & 0 \\
     r & 0 & r & 0  \\
     0 & 0 & 0 & r
 \end{pmatrix}$ & 1 & AMDS  \\ [7pt]

 &\hspace{2cm} $\begin{pmatrix}
     r & r & 0 & 0 \\
     r & 0 & r & 0  \\
     2r & 0 & 0 & r
 \end{pmatrix}$ & 2 & MDS \\ [15pt]
 & \hspace{2.15cm}$\begin{pmatrix}
     0 & r & 0 & 0 \\
     0& 0 & r & 0  \\
     0 & 0 & 0 & r
 \end{pmatrix}$ & 1 & AMDS \\ [15pt]
 & \hspace{2.4cm} $rI_4$ & 1 & MDS \\ [5pt]
\hline

&\hspace{2cm} $\begin{pmatrix}
    r & r & r & r &r
\end{pmatrix}$ & 5 & MDS \\ [5pt]
& \hspace{2cm} $\begin{pmatrix}
    r & r & 2r & 0 &r
\end{pmatrix}$ & 4 & AMDS \\ [5pt]
  & \hspace{2cm}$\begin{pmatrix}
    r & 0 & 2r & 0 & r \\
    0 & r & 2r & 2r & 0
\end{pmatrix}$ & 3 & AMDS \\ [10pt]

 & \hspace{1.9cm} $\begin{pmatrix}
    r & 0 & r & 0 & 0 \\
    r & 0 & 0 & r & 0 \\
    2r & 0 & 0 & 0 & r
\end{pmatrix}$ & 2 & AMDS\\ [15pt]

$5$ & \hspace{2cm}$\begin{pmatrix}
    r & r & r & 0 & 0 \\
    0 & r & 0 & r & 0 \\
    2r & 0 & 0 & 0 & r
\end{pmatrix}$ & 2 & AMDS \\ [15pt]
& \hspace{2cm}$\begin{pmatrix}
    0 & r & r & 0 & 0 \\
    r & r & 0 & r & 0 \\
    0 & 2r & 0 & 0 & r
\end{pmatrix}$ & 2 & AMDS \\ [15pt]

&\hspace{2cm}$\begin{pmatrix}
    2r & r & 0 & 0 & 0 \\
    2r & 0 & r & 0 & 0 \\
    2r & 0 & 0 & r & 0 \\
    2r & 0 & 0 & 0 & r
\end{pmatrix}$ & 2 & MDS \\ [7pt]
& \hspace{2cm}$\begin{pmatrix}
    2r & r & 0 & 0 & 0 \\
    r & 0 & r & 0 & 0 \\
    0 & 0 & 0 & r & 0 \\
    2r & 0 & 0 & 0 & r
\end{pmatrix}$ & 1 & AMDS \\ [20pt]
& \hspace{1.95cm} $\begin{pmatrix}
    0 & r & 0 & 0 & 0 \\
    0 & 0 & r & 0 & 0 \\
    0 & 0 & 0 & r & 0 \\
    0 & 0 & 0 & 0 & r
\end{pmatrix}$ & 1 & AMDS \\ [20pt]
& \hspace{2cm}$\begin{pmatrix}
    2r & r & 0 & 0 & 0 \\
    0 & 0 & r & 0 & 0 \\
    0 & 0 & 0 & r & 0 \\
    0 & 0 & 0 & 0 & r
\end{pmatrix}$ & 1 & AMDS  \\ [15pt]
 & \hspace{2.6cm} $rI_5$ & 1 & MDS \\ [5pt]
 \hline

 & \hspace{2cm}$\begin{pmatrix}
     r & r & 2r & 0 & 2r & r
 \end{pmatrix}$ & 5 & AMDS \\ [5pt]
 & \hspace{2cm}$\begin{pmatrix}
     r & 0 & r & r & r & r \\
     0 & r & 2r & 2r & r & r
 \end{pmatrix}$ & 4 & AMDS \\ [10pt]

 & \hspace{2cm}$\begin{pmatrix}
     r & 0 & 0 & 2r & 0 & r \\
     0 & r & 0 & 2r & 2r & 0 \\
     0 & 0 & r & 0 & 2r & r
 \end{pmatrix}$ & 3 & AMDS \\ [15pt]
& \hspace{2cm}$\begin{pmatrix}
     r & 0 & 0 & 2r & 0 & r \\
     0 & r & 0 & 2r & 2r & 0 \\
     0 & 0 & r & 2r & 2r & r
 \end{pmatrix}$ & 3 & AMDS  \\ [15pt]

& \hspace{2cm}$\begin{pmatrix}
    2r & 0 & r & 0 & 0 & 0 \\
    0 & 2r & 0 & r & 0 & 0 \\
    0 & r & 0 & 0 & r & 0 \\
    0 & 2r & 0 & 0 & 0 & r
\end{pmatrix}$ & 2 & AMDS \\ [20pt]
& \hspace{2cm}$\begin{pmatrix}
    2r & r & r & 0 & 0 & 0 \\
    2r & r & 0 & r & 0 & 0 \\
    2r & 2r & 0 & 0 & r & 0 \\
    2r & 2r & 0 & 0 & 0 & r
\end{pmatrix}$ & 2 & AMDS \\ [20pt]
$6$ & \hspace{2cm}$\begin{pmatrix}
    2r & 0 & r & 0 & 0 & 0 \\
    2r & 0 & 0 & r & 0 & 0 \\
    2r & 0 & 0 & 0 & r & 0 \\
    2r & 0 & 0 & 0 & 0 & r
\end{pmatrix}$ & 2 & AMDS  \\ [20pt]
& \hspace{2cm}$\begin{pmatrix}
    2r & r & r & 0 & 0 & 0 \\
    0 & r & 0 & r & 0 & 0 \\
    0 & 2r & 0 & 0 & r & 0 \\
    r & 2r & 0 & 0 & 0 & r
\end{pmatrix}$ & 2 & AMDS \\ [20pt]

& \hspace{2.1cm}$\begin{pmatrix}
    0 & r & 0 & 0 & 0 & 0 \\
    0 & 0 & r & 0 & 0 & 0 \\
    r & 0 & 0 & r & 0 & 0 \\
    0 & 0 & 0 & 0 & r & 0 \\
    0  & 0 & 0 & 0 & 0 & r
\end{pmatrix}$ & 1 & AMDS \\ [25pt]
&\hspace{1.9cm} $\begin{pmatrix}
    0 & r & 0 & 0 & 0 & 0 \\
    r & 0 & r & 0 & 0 & 0 \\
    2r & 0 & 0 & r & 0 & 0 \\
    2r & 0 & 0 & 0 & r & 0 \\
    0  & 0 & 0 & 0 & 0 & r
\end{pmatrix}$ & 1 & AMDS \\ [25pt]
& \hspace{1.9cm} $\begin{pmatrix}
    2r & r & 0 & 0 & 0 & 0 \\
    r & 0 & r & 0 & 0 & 0 \\
    0 & 0 & 0 & r & 0 & 0 \\
    r & 0 & 0 & 0 & r & 0 \\
    2r  & 0 & 0 & 0 & 0 & r
\end{pmatrix}$ & 1 & AMDS  \\ [25pt]

&\hspace{1.95cm} $\begin{pmatrix}
    0 & r & 0 & 0 & 0 & 0 \\
    0 & 0 & r & 0 & 0 & 0 \\
    0 & 0 & 0 & r & 0 & 0 \\
    0 & 0 & 0 & 0 & r & 0 \\
    0  & 0 & 0 & 0 & 0 & r
\end{pmatrix}$ & 1 & AMDS \\ [5pt]
& \hspace{2.6cm} $rI_6$ & 1 & MDS \\ [3pt]

\hline

\end{longtable}

\section{Self-dual codes}
This Section focuses on the study of left, right and two-sided self-dual codes over $E_p$. Firstly, we study left self-dual codes over $E_p$ and classify these codes which are MDS and AMDS over $E_2$ and $E_3$ for lengths up to $12$. Then, we study right self-dual codes over $E_p$ and prove the non-existence of right self-dual MDS codes. Also, we show that the right self-dual AMDS codes over $E_p$ exist only for the code length $2$. Finally, we study MDS and AMDS two-sided self-dual codes over $E_p$ and classify these codes over $E_2$ and $E_3$ for lengths up to $6$ and $4$, respectively.

\begin{definition}[Left and right self-dual codes]
    If an $E_p$-linear code $C$ satisfies $C=C^{\perp_L}$ (resp. $C=C^{\perp_R}$), it is known as a left (resp. right) self-dual code.
\end{definition}

\begin{definition}[Quasi-self-dual and self-dual codes]
 If an $E_p$-linear code $C$ of length $n$ satisfies $C\subseteq C^{\perp_L}\cap C^{\perp_R}$ and $|C|=p^n$, it is known as a quasi-self-dual (QSD) code. Further, if $C=C^{\perp_L}\cap C^{\perp_R}$, the code $C$  refers to a self-dual code.
\end{definition}

The next result gives a condition under which an $E_p$-linear code $C$ becomes left self-dual.
\begin{theorem}\label{thm25}
 An $E_p$-linear code $C$ is a left self-dual code over $E_p$ if and only if $C$ is free and its residue code is a self-dual code over $\mathbb{F}_p$.
\end{theorem}
\begin{proof}
    Let $C$ be a left self-dual code over $E_p$. Then, $C$ is free, by Lemma $5$ from \cite{Shi21}. We have $Res(C)=Res(C^{\perp_L})$, as $C=C^{\perp_L}$. On the other hand, $Res(C^{\perp_L})=Res(C)^\perp$  by Lemma $8$ from \cite{Shi21}. This implies that $Res(C)=Res(C)^\perp$. Therefore, the residue code is self-dual.

    Conversely, suppose that the $E_p$-linear code $C$ is free and its residue code is self-dual. Then, by Theorem \ref{thm7},
    \begin{align*}
        C^{\perp_L} & = rRes(C^{\perp_L}) \oplus tTor(C^{\perp_L}) \\
        & =rRes(C^{\perp_L}) \oplus tRes(C^{\perp_L})  \\
        &= rRes(C)^\perp \oplus tRes(C)^\perp    \\
        &=  rRes(C) \oplus tRes(C)  \\
        &=  rRes(C) \oplus tTor(C) \\
        &= C.
        \end{align*}
        Thus, $C$ is left self-dual.
\end{proof}

The following result gives a necessary and sufficient condition under which a left self-dual $E_p$-code becomes MDS or AMDS.
\begin{theorem}\label{thm26}
        A left self-dual $E_p$-code $C$ is MDS (resp. AMDS)  if and only if its residue code is an MDS (resp. AMDS) code over $\mathbb{F}_p$.
\end{theorem}
\begin{proof}
    By Theorem \ref{thm25}, a left self-dual $E_p$-code is always free. Therefore, by Theorem \ref{thm13}, a left self-dual $E_p$-code $C$ is MDS (resp. AMDS)  if and only if its residue code is an MDS (resp. AMDS) code over $\mathbb{F}_p$.
\end{proof}

The next result enumerates monomial inequivalent MDS and AMDS left self-dual codes over $E_p$.

\begin{theorem}
    The number of monomial inequivalent MDS (resp. AMDS) left self-dual codes of length $n$ over $E_p$ with minimum distance $d$ and the number of monomial inequivalent MDS (resp. AMDS) self-dual codes over $\mathbb{F}_p$ with the same length and minimum distance are equal.
\end{theorem}
\begin{proof}
    By Theorem \ref{thm25}, a left self-dual code $C$ over $E_p$ is always free and its residue code is a self-dual code. Then the result holds directly from Theorems \ref{thm 5E} and \ref{thm26}.
\end{proof}

 Our next step is to classify MDS and AMDS left self-dual codes over $E_2$ and $E_3$ for some specific lengths. Theorem \ref{thm26} can be used to classify MDS and AMDS left self-dual codes over $E_p$. We refer \cite{database,Mallows,Pless} for inequivalent binary and ternary self-dual codes and use MAGMA computer algebra system \cite{Magma} to find MDS and AMDS binary and ternary self-dual codes. Then, we use  Theorems \ref{thm2} and \ref{thm26} to construct monomial inequivalent MDS and AMDS left self-dual codes over $E_2$ and $E_3$ for lengths up to $12$ and list them in Tables \ref{Tab7} and \ref{Tab8}, respectively. The length (up to $12$) is absent in the tables if there exists no MDS or AMDS left self-dual code of that length.   \\

\begin{longtable}{|c|c|c|c|}
\caption{\label{Tab7} Monomial inequivalent MDS and AMDS left self-dual   codes over $E_2$.} \\
\hline
$n$ &~~~~~~~~ Generator Matrix & Minimum  & Remark \\ && Distance &  \\

 \hline

 $2$ & $\begin{pmatrix}
     r& r
 \end{pmatrix}$  & 2 & MDS\\ [5pt]

 \hline

  $4$ & $\begin{pmatrix}
     r& r & 0 & 0 \\
     0 & 0 & r & r
 \end{pmatrix}$  & 2 & AMDS\\ [10pt]
 \hline

  $8$ & $\begin{pmatrix}
     r& 0 & 0 & 0 & 0 & r & r & r \\
     0 & r & 0 & 0 & r & 0 & r & r\\
     0 & 0 & 0 & r & r & r & 0 & r \\
     0 & 0 & 0 & r & r & r & r & r
 \end{pmatrix}$  & 4 & AMDS\\ [10pt]

\hline

\end{longtable}

\begin{longtable}{|c|c|c|c|}
\caption{\label{Tab8} Monomial inequivalent MDS and AMDS left  self-dual   codes over $E_3$.}\\
\hline
$n$ &~~~~~~~~ Generator Matrix & Minimum  & Remark \\ && Distance &  \\

 \hline

  $4$ & $\begin{pmatrix}
     r & 0 & r & r \\
     0 & r & r & 2r
 \end{pmatrix}$  & 3 & MDS\\ [10pt]
 \hline

 \setcounter{MaxMatrixCols}{12}
  $12$ & $\begin{pmatrix}

     r & 0 & 0 & 0 & 0 & 0 & 0  & r  & r  & r & r  & r \\
     0 & r & 0 & 0 & 0 & 0 & 2r & 0  & r  & 2r & 2r & r \\
     0 & 0 & r & 0 & 0 & 0 & 2r & r  & 0  & r  & 2r & 2r  \\
     0 & 0 & 0 & r & 0 & 0 & 2r & 2r & r  & 0  & r  & 2r \\
     0 & 0 & 0 & 0 & r & 0 & 2r & 2r & 2r & r  & 0  & r \\
     0 & 0 & 0 & 0 & 0 & r & 2r & r  & 2r & 2r & r  & 0
 \end{pmatrix}$  & 6 & AMDS \\ [20pt]

\hline

\end{longtable}

In the next result, we investigate MDS and AMDS right self-dual codes over $E_p$.

\begin{theorem}
    No MDS right self-dual code exists over the ring $ E_p$. Moreover, a right self-dual code $C$ over $E_p$ is AMDS if and only if its length is $2$.
\end{theorem}
\begin{proof}
  By (iv) of Theorem 2 in \cite{Alah24}, an $E_p$-code $C$ of length $n$ is right self-dual if and only if $C=t\mathbb{F}_p^n$. Then,
  $$d=d(C)=d(t\mathbb{F}_p^n)=d(\mathbb{F}_p^n)=1.$$
  Also, $$|C|=|t\mathbb{F}_p^n|=|\mathbb{F}_p^n|=p^n.$$
  Hence, $|E_p|^{n-d+1}=({p^2})^{n-1+1}=p^{2n}$ and $|E_p|^{n-d}=({p^2})^{n-1}=p^{2n-2}$. This shows that $C$ can never be an MDS right self-dual code over $E_p$, as $|C|\neq |E_p|^{n-d+1}$. On the other hand, suppose that the right self-dual $E_p$-code $C$ is AMDS. Then, $p^n=p^{2n-2}$ implies that $n=2n-2$ and so $n=2$. Conversely, if the right self-dual code $C$ is of length $2$, it is AMDS. Thus, the code $C$ is AMDS if and only if its length is $2$.
\end{proof}

Next, we study self-dual  MDS and AMDS codes over $E_p$. In Remark $2$ of \cite{Alah23}, authors have proved that the concepts of QSD codes and self-dual codes over $E_2$ coincide. This result is also true for codes over $E_p$. The proof follows the same approach as in \cite{Alah23}. We use this result in the next theorem to show that the odd lengths MDS and AMDS self-dual codes do not exist over $E_p$.
 \begin{theorem}\label{thm21a}
 Any MDS or AMDS self-dual code over  $E_p$  must have an even length. Moreover, $m_2$ is also even.
 \end{theorem}
 \begin{proof}
     Let $C$ be a self-dual code of length $n$ over $E_p$ with minimum distance $d$. Since all the self-dual codes over $E_p$ are also QSD, we have $|C|=p^n$. Now, if $C$ is MDS, then $|C|=p^n=|E_p|^{n-d+1}$ implies that $n=2n-2d+2$. Hence, $n=2(d-1)$. On the other hand, if $C$ is AMDS, then  $|C|=p^n=|E_p|^{n-d}$ implies that $n=2n-2d$. Therefore, $n=2d$. Thus, from both scenarios, we conclude that any MDS or AMDS self-dual code over $E_p$  must have an even length. Further, $|C|=p^{2m_1+m_2}=p^n$ implies that $2m_1+m_2=n$, and so $m_2=n-2m_1$. Therefore, $m_2$ is also even.
 \end{proof}

The next result shows that there is no MDS self-dual code over $E_p$ with minimum distance $1$.
 \begin{theorem}
     There exist no MDS self-dual code over $E_p$ with $d(C)=1$.
 \end{theorem}
 \begin{proof}
   Let $C$ be an MDS self-dual code of length $n$ over $E_p$ with minimum distance $d$. Then, by Theorem \ref{thm21a}, $n=2(d-1)$. This shows that $d$ should always be greater than $1$, and hence the result follows.
 \end{proof}

Our next purpose is to classify self-dual MDS and AMDS codes over $E_2$ and $E_3$. QSD codes over $E_2$ are classified in \cite{Alah22} for lengths up to $6$. We use this classification and MAGMA computer algebra system \cite{Magma} to classify self-dual MDS and AMDS  codes over $E_2$ for lengths up to $6$ and list them in Table \ref{Tab9}. Moreover, authors in \cite{Alah22} have shown that for a QSD code $C$ over $E_2$, $Tor(C)=Res(C)^\perp$ and  $Res(C)$ is a self-orthogonal code. This result is also true for codes over $E_p$. We use this result and Theorem \ref{thm7} to construct monomial inequivalent self-dual codes over $E_3$ and use MAGMA computer algebra system \cite{Magma} to classify MDS and AMDS self-dual codes over $E_3$ for lengths up to $4$ and list them in Table \ref{Tab10}. The lengths (up to $6$  and $4$, resp.) are absent in the tables if no MDS or AMDS self-dual code exists for that length.

\begin{longtable}{|c|p{6cm}|c|c|}
\caption{\label{Tab9} Monomial inequivalent MDS and AMDS self-dual codes over $E_2$.}\\
\hline
 $n$ &~~~~~~~~ Generator Matrix & Minimum  & Remark \\ && Distance &  \\

 \hline
 & \hspace{1.8cm}$\begin{pmatrix}
     r & ~ r
 \end{pmatrix}$ & $2$ & MDS \\
 $2$ & \hspace{2cm}$tI_2$ & $1$ & AMDS \\ [5pt]
 \hline

 & \hspace{1.5cm}$\begin{pmatrix}
     r & 0 & r & 0 \\
     0 & r & 0 & r

 \end{pmatrix}$  & 2 & AMDS\\ [10pt]
  $4$ & \hspace{1.5cm}$\begin{pmatrix}
     t & 0 & 0 & t \\
     r & r & r & r \\
     0 & t & 0 & t \\
     0 & 0 & t & t
 \end{pmatrix}$  & 2 & AMDS\\ [20pt]

\hline

\end{longtable}

\begin{longtable}{|c|p{6cm}|c|c|}
\caption{\label{Tab10} Monomial inequivalent MDS and AMDS self-dual codes over $E_3$.}\\
\hline
 $n$ &~~~~~~~~ Generator Matrix & Minimum  & Remark \\ && Distance &  \\

 \hline
 $2$ & \hspace{2cm}$tI_2$ & $1$ & AMDS \\
 \hline

$4$ & \hspace{1.5cm}$\begin{pmatrix}
     r & 0 & 2r & 2r \\
     0 & r & 2r & r

 \end{pmatrix}$  & 3 & MDS\\ [10pt]

\hline

\end{longtable}

\section{Conclusion}
This paper presented the study of LCD and self-dual codes over a non-unital noncommutative ring $E_p$. Firstly, the enumeration of monomial inequivalent LCD codes over $E_2$ and $E_3$ with a fixed code length and minimum distance is carried out for lengths up to $13$ and $10$, respectively. Then, the classification of MDS and AMDS  LCD codes over $E_2$ and $E_3$ is presented for lengths up to $6$. Further, we have introduced MDS and AMDS left self-dual codes over $E_p$ and classified these codes over $E_2$ and $E_3$ for lengths up to $12$. Also, we have shown that MDS right self-dual codes do not exist over $E_p$. Finally, self-dual codes have been introduced over the ring $E_p$, and then MDS and AMDS self-dual codes over $E_2$ and $E_3$ are classified for smaller lengths.
\section*{Acknowledgement}
The first author thanks the Council of Scientific \& Industrial Research (under grant No. 09/1023(16098)/2022-EMR-I), Govt. of India, for financial support.
\section*{Declarations}
\textbf{Use of AI tools}:
The authors confirm that no Artificial Intelligence (AI) tools were used for the preparation of this manuscript.\par
\textbf{Competing interests}: The authors declare that they have no conflict of interest related to the publication of this manuscript.\par

\textbf{Statement on Data Availability}: The authors confirm that this manuscript encompasses all the data used to support the findings of this study. For any essential clarifications, requests can be directed to the corresponding author.

\end{document}